\documentclass[graybox]{svmult}

\usepackage{amsmath, amsfonts, amssymb}

\usepackage{type1cm}        
\usepackage{makeidx}         
\usepackage{graphicx}        
\usepackage{multicol}        
\usepackage[bottom]{footmisc}

\newcommand{\al}{\alpha}
\newcommand{\bt}{\beta}

\newcommand{\no}{\nonumber}
\newcommand{\bs}{\begin{small}}
\newcommand{\es}{\end{small}}
\newcommand{\be}{\begin{equation}}
\newcommand{\ee}{\end{equation}}
\newcommand{\bea}{\begin{eqnarray}}
\newcommand{\eea}{\end{eqnarray}}

\makeindex             

\begin{document}

\title*{Global and Local Scaling Limits for Linear Eigenvalue Statistics of Jacobi $\beta$-Ensembles}
\titlerunning{Linear Eigenvalue Statistics of Jacobi $\beta$-Ensembles}
\author{Chao Min and Yang Chen}
\institute{Chao Min \at Chao Min, School of Mathematical Sciences, Huaqiao University, Quanzhou 362021, China, \email{chaomin@hqu.edu.cn}
\and Yang Chen \at Yang Chen, Department of Mathematics, Faculty of Science and Technology, University of Macau, Macau, China \email{yangbrookchen@yahoo.co.uk}}
%
%
\maketitle

\vspace*{-1cm}
\centerline{\emph{Dedicated to the memory of Harold Widom}}

\vspace{1cm}

\abstract*{We study the moment-generating functions (MGF) for linear eigenvalue statistics of Jacobi unitary, symplectic and orthogonal ensembles. By expressing the MGF as Fredholm determinants of kernels of finite rank, we show that the mean and variance of the suitably scaled linear statistics in these Jacobi ensembles are related to the sine kernel in the bulk of the spectrum, whereas they are related to the Bessel kernel at the (hard) edge of the spectrum. The relation between the Jacobi symplectic/orthogonal ensemble (JSE/JOE) and the Jacobi unitary ensemble (JUE) is also established.}

\abstract{We study the moment-generating functions (MGF) for linear eigenvalue statistics of Jacobi unitary, symplectic and orthogonal ensembles. By expressing the MGF as Fredholm determinants of kernels of finite rank, we show that the mean and variance of the suitably scaled linear statistics in these Jacobi ensembles are related to the sine kernel in the bulk of the spectrum, whereas they are related to the Bessel kernel at the (hard) edge of the spectrum. The relation between the Jacobi symplectic/orthogonal ensemble (JSE/JOE) and the Jacobi unitary ensemble (JUE) is also established.}



\section{Introduction}
In random matrix theory (RMT), the joint probability density function for the (real) eigenvalues $\{x_j\}_{j=1}^{N}$ of $N\times N$ Hermitian matrices from a matrix ensemble
is given by \cite{Mehta}
\be\label{jpdf}
P_{N}^{(\beta)}(x_{1},x_{2},\ldots,x_{N})=\frac{1}{Z_{N}}\prod_{1\leq j<k\leq N}\left|x_{j}-x_{k}\right|^{\beta}\prod_{j=1}^{N}w(x_{j}),
\ee
where $\bt=1, 2$ and $4$ (the Dyson index) correspond to the orthogonal, unitary and symplectic ensembles respectively, $w(x)$ is a weight function and $Z_{N}$ is a normalization constant.
If $w(x)=\mathrm{e}^{-x^2},\; x\in \mathbb{R}$ and $w(x)=x^{\al}\mathrm{e}^{-x},\;  x\in \mathbb{R}^{+},\al>-1$, these are the Gaussian $\bt$-ensembles (G$\bt$E) and Laguerre $\bt$-ensembles (L$\bt$E). See also \cite{Widom} on the relation between orthogonal, symplectic and unitary ensembles.

Linear statistics is an important research object in RMT and has various applications; see, e.g., \cite{Bao,Beenakker1,Beenakker2,Chen1998,Chen1994,Cunden,GMT,KL,Li,Vivo}. In previous works \cite{Min201601,Min2020}, the authors studied the large $N$ asymptotics for the moment-generating functions (MGF) of the suitably scaled linear statistics in G$\bt$E and L$\bt$E, from which the mean and variance of the linear statistics are derived. In the present paper, we focus on the problem in Jacobi $\bt$-ensembles (J$\bt$E). In this case, the weight function is $w(x)=(1-x)^a(1+x)^b,\; x\in[-1,1],\; a, b>-1$.

The MGF of the linear statistics $\sum_{j=1}^{N}F(x_{j})$ in J$\bt$E is given by the mathematical expectation with respect to the joint probability density function (\ref{jpdf}),
\be\label{mgf}
\mathbb{E}\left(\mathrm{e}^{-\lambda\sum_{j=1}^{N}F(x_j)}\right)
=\frac{\int_{[-1,1]^{N}}\prod_{1\leq j<k\leq N}\left|x_{j}-x_{k}\right|^{\beta}\prod_{j=1}^{N}w(x_{j})\mathrm{e}^{-\lambda F(x_j)}dx_{j}}{\int_{[-1,1]^{N}}\prod_{1\leq j<k\leq N}\left|x_{j}-x_{k}\right|^{\beta}\prod_{j=1}^{N}w(x_{j})dx_{j}},
\ee
where $\lambda$ is a parameter and $F(\cdot)$ is a sufficiently well-behaved function to make the integral well-defined.
Similarly as in \cite{Min201601,Min2020}, we write the right-hand side of (\ref{mgf}) in the form
\be\label{gnb}
G_{N}^{(\beta)}(f):=\frac{\int_{[-1,1]^{N}}\prod_{1\leq j<k\leq N}\left|x_{j}-x_{k}\right|^{\beta}
\prod_{j=1}^{N}w(x_{j})\left(1+f(x_{j})\right)dx_{j}}{\int_{[-1,1]^{N}}\prod_{1\leq j<k\leq N}\left|x_{j}-x_{k}\right|^{\beta}\prod_{j=1}^{N}w(x_{j})dx_{j}},
\ee
where
\be\label{re}
f(x)=\mathrm{e}^{-\lambda F(x)}-1.
\ee
The denominator in (\ref{mgf}) or (\ref{gnb}) is known as Selberg's integral, which has closed form expression \cite[(17.6.1)]{Mehta}. We are interested in the large $N$ asymptotics of the MGF. It is well known that the distributions of linear statistics in random matrix ensembles are Gaussian; see, e.g., \cite{Chen1994,Politzer}.

We first consider the $\bt=2$ case, which is the simplest among the three cases. From the well-known result of Tracy and Widom \cite{Tracy1998} by expressing $G_{N}^{(2)}(f)$ as a Fredholm determinant, we obtain its large $N$ asymptotics. With the relation of $f(x)$ and $F(x)$, we compute the mean and variance of linear statistics $\sum_{j=1}^{N}F(x_{j})$ in the bulk of the spectrum and at the edge respectively. It can be seen that in the bulk of the spectrum the results are related to the sine kernel, while at the edge they are related to the Bessel kernel. The mean and variance of the linear statistics in unitary ensembles have been studied a lot; see, e.g., \cite{Basor1997,Basor1993,Chen1998,Min2016}. So the mean and variance in the $\bt=2$ case can also be obtained by using other approaches. Our main goal of this paper is to obtain the mean and variance in the $\bt=4$ and $\bt=1$ cases for Jacobi ensembles. We show the results of the $\bt=2$ case for reference and apply the method to the $\bt=4$ and $\bt=1$ cases.

For the $\bt=4$ case, we apply the previous results for general weight functions in \cite{Min201601,Min2020} to the Jacobi weight. By making use of the skew orthogonal polynomials for the Jacobi weight \cite{Adler2000}, we express $G_{N}^{(4)}(f)$ as a Fredholm determinant involving the Christoffel-Darboux kernel. The large $N$ asymptotics of $G_{N}^{(4)}(f)$ is derived by using the trace-log expansions. The mean and variance of the scaled linear statistics $\sum_{j=1}^{N}F(x_{j})$ then follows and the relation between the $\bt=4$ case and $\bt=2$ case is built.

The $\bt=1$ case is more difficult to deal with, and we only consider the case when $N$ is even. Usually in this situation the weight is taken to be the square root of the weight considered in the $\bt=2$ case, so we let $w(x)=(1-x)^{a/2}(1+x)^{b/2},\; x\in[-1,1],\; a, b>-2$. The following development is similar to the $\bt=4$ case, but with more complicated computations. Finally we obtain the mean and variance of the scaled linear statistics $\sum_{j=1}^{N}F(x_{j})$ and establish the relation between the $\bt=1$ case and $\bt=2$ case. Note that as in the $\bt=2$ case we also consider the $\bt=4$ and $\bt=1$ cases in the bulk of the spectrum and at the edge, and the results are related to the sine kernel and Bessel kernel, respectively.

We would like to point out that in this paper some calculations on the asymptotics are heuristic. To be specific, we always substitute the asymptotic expressions of the traces into the trace-log expansions for the MGF and do not care much about the error terms. So the error estimates in the asymptotic analysis should be made more precisely, such as the errors in the mean and variance formulas of the scaled linear statistics obtained in the following sections.

\section{Jacobi Unitary Ensemble (JUE)}
In this section, we consider the $\bt=2$ case, which is the simplest case and provides a comparison to the $\bt=4$ and $\bt=1$ cases.
\subsection{Finite $N$ Case for the MGF in JUE}
Recall that the weight function is $w(x)=(1-x)^{a}(1+x)^{b},\;x\in[-1,1],\;a, b>-1$. Let $\{\varphi_{j}(x)\}_{j=0}^{\infty}$ be the sequence obtained by orthonormalizing the sequence $\{x^{j}(1-x)^{a/2}(1+x)^{b/2}\}$ in $L^2[-1,1]$ and
\be\label{kn2}
K_{N}^{(2)}(x,y):=\sum_{j=0}^{N-1}\varphi_{j}(x)\varphi_{j}(y).
\ee
In fact,
$$
\varphi_{j}(x)=\frac{1}{\sqrt{h_j^{(a,b)}}}P_j^{(a,b)}(x)(1-x)^{a/2}(1+x)^{b/2},
$$
where $P_j^{(a,b)}(x)$ is the Jacobi polynomial of degree $j$ with the orthogonality \cite{Ismail,Szego}
$$
\int_{-1}^{1}P_j^{(a,b)}(x)P_k^{(a,b)}(x)(1-x)^a(1+x)^b dx=h_j^{(a,b)}\delta_{jk},\qquad j, k=0,1,2,\ldots
$$
and
$$
h_j^{(a,b)}=\frac{2^{a+b+1}\Gamma(j+a+1)\Gamma(j+b+1)}{j!(2j+a+b+1)\Gamma(j+a+b+1)}.
$$

Tracy and Widom \cite{Tracy1998} proved that $G_{N}^{(2)}(f)$ can be expressed as a Fredholm determinant
$$
G_{N}^{(2)}(f)=\det\left(I+K_{N}^{(2)}f\right),
$$
where $K_{N}^{(2)}$ is the operator on $L^2[-1,1]$ with kernel $K_{N}^{(2)}(x,y)$ given by (\ref{kn2}),
and $f$ denotes the operator of multiplication by $f$. In addition, it is well known that
\bea\label{logdet}
&&\log\det\left(I+K_{N}^{(2)}f\right)
=\mathrm{Tr}\log\left(I+K_{N}^{(2)}f\right)\no\\
&&=\mathrm{Tr}\:K_{N}^{(2)}f-\frac{1}{2}\mathrm{Tr}\left(K_{N}^{(2)}f\right)^{2}+\frac{1}{3}\mathrm{Tr}\left(K_{N}^{(2)}f\right)^{3}-\cdots.
\eea
This formula will help us to analyze the large $N$ asymptotics of $G_{N}^{(2)}(f)$ in the following subsections.

\subsection{Scaling in the Bulk of the Spectrum in JUE}
In this subsection, we study the large $N$ asymptotics of $G_{N}^{(2)}(f)$ in the bulk of the spectrum for the JUE, and obtain the mean and variance of the suitably scaled linear statistics.
We state a theorem before our discussion.
\begin{theorem}\label{jue1}
For $x, y\in\mathbb{R}$, we have as $N\rightarrow\infty$,
$$
\frac{1}{N}K_{N}^{(2)}\left(\frac{x}{N},\frac{y}{N}\right)
=K_{\mathrm{sine}}(x,y)+O(N^{-1}),
$$
where $K_{\mathrm{sine}}(x,y)$ is the sine kernel defined by
\be\label{sine}
K_{\mathrm{sine}}(x,y):=\frac{\mathrm{sin}(x-y)}{\pi(x-y)}.
\ee
The error term is uniform for $x$ and $y$ in compact subsets of $\mathbb{R}$.
\end{theorem}
\begin{proof}
By using the Christoffel-Darboux formula, we have
\bea\label{cd}
K_{N}^{(2)}(x,y)&=&\frac{\Gamma(N+1)\Gamma(N+a+b+1)}{2^{a+b}(2N+a+b)\Gamma(N+a)\Gamma(N+b)}\no\\
&\times&\frac{P_N^{(a,b)}(x)P_{N-1}^{(a,b)}(y)-P_{N-1}^{(a,b)}(x)P_N^{(a,b)}(y)}{x-y}
\nonumber\\
&\times&(1-x)^{a/2}(1+x)^{b/2}(1-y)^{a/2}(1+y)^{b/2}.
\eea
Taking advantage of the large $n$ asymptotic formula of the Jacobi polynomials \cite[p. 196]{Szego}
\bea\label{asy}
P_n^{(a,b)}(\cos \theta)&=&\frac{1}{\sqrt{\pi n}}\left(\sin\frac{\theta}{2}\right)^{-a-\frac{1}{2}}\left(\cos\frac{\theta}{2}\right)^{-b-\frac{1}{2}}\no\\
&\times&\cos\left[\left(n+\frac{a+b+1}{2}\right)\theta-\frac{\pi}{2}\left(a+\frac{1}{2}\right)\right]+O(n^{-3/2}),
\eea
where $0<\theta<\pi$, we find that $K_{N}^{(2)}(\cos \theta,\cos \phi)$ equals
\bs
\bea
&&\frac{\Gamma(N+1)\Gamma(N+a+b+1)}{2^{a+b}(2N+a+b)\Gamma(N+a)\Gamma(N+b)(\cos \theta-\cos \phi)}\Bigg\{\frac{2^{a+b+1}}{\pi\sqrt{N(N-1)\sin\theta\sin\phi}}\nonumber\\
&\times&\bigg[\cos\left(\left(N+\frac{a+b+1}{2}\right)\theta-\frac{\pi}{2}\left(a+\frac{1}{2}\right)\right)
\cos\left(\left(N+\frac{a+b-1}{2}\right)\phi-\frac{\pi}{2}\left(a+\frac{1}{2}\right)\right)\nonumber\\
&-&\cos\left(\left(N+\frac{a+b-1}{2}\right)\theta-\frac{\pi}{2}\left(a+\frac{1}{2}\right)\right)
\cos\left(\left(N+\frac{a+b+1}{2}\right)\phi-\frac{\pi}{2}\left(a+\frac{1}{2}\right)\right)\bigg]\nonumber\\
&+&O(N^{-2})\Bigg\},\qquad 0<\theta,\; \phi<\pi. \nonumber
\eea
\es%
The above error terms are uniform for $\theta$ and $\phi$ in compact subsets of $(0,\pi)$.
Note that the expression in the square brackets $[\cdots]$ can be written in the form
\bs
\bea
&&2\bigg[\cos\left(\left(N+\frac{a+b}{2}\right)\theta-\frac{\pi}{2}\left(a+\frac{1}{2}\right)\right)
\sin\left(\left(N+\frac{a+b}{2}\right)\phi-\frac{\pi}{2}\left(a+\frac{1}{2}\right)\right)\cos\frac{\theta}{2}\sin\frac{\phi}{2}\nonumber\\
&-&\sin\left(\left(N+\frac{a+b}{2}\right)\theta-\frac{\pi}{2}\left(a+\frac{1}{2}\right)\right)
\cos\left(\left(N+\frac{a+b}{2}\right)\phi-\frac{\pi}{2}\left(a+\frac{1}{2}\right)\right)\sin\frac{\theta}{2}\cos\frac{\phi}{2}\bigg].\nonumber
\eea
\es
Replacing $\cos \theta$ and $\cos \phi$ by $x/N$ and $y/N$ respectively and taking a large $N$ limit, we establish the theorem with the aid of Stirling's formula. See also \cite{Fox,Nagao}.
\end{proof}
\begin{remark}
When $x=y$, we have as $N\rightarrow\infty$,
$$
\frac{1}{N}K_{N}^{(2)}\left(\frac{x}{N},\frac{x}{N}\right)
=K_{\mathrm{sine}}(x,x)+O(N^{-1}),
$$
where
$$
K_{\mathrm{sine}}(x,x)=\frac{1}{\pi}.
$$
The error term is uniform for $x$ in compact subsets of $\mathbb{R}$.
\end{remark}

Using Theorem \ref{jue1}, we compute (\ref{logdet}) term by term as $N\rightarrow\infty$, and we change $f(x)$ to $f(Nx)$ in the computations. The first term is
\bea
\mathrm{Tr}K_{N}^{(2)}f
&=&\int_{-1}^{1}K_{N}^{(2)}(x,x)f(Nx)dx\nonumber\\
&=&\int_{-N}^{N}\frac{1}{N}K_{N}^{(2)}
\left(\frac{x}{N},\frac{x}{N}\right)f(x)dx\nonumber\\
&=&\int_{-\infty}^{\infty}K_{\mathrm{sine}}(x,x)f(x)dx+O(N^{-1}),\qquad N\rightarrow\infty.\nonumber
\eea
\begin{remark}
We assume that $f(\cdot)$ is a continuous real-valued function belonging to $L^{1}(\mathbb{R})$ and vanishes at $\pm\infty$.
\end{remark}
The second term gives
\bea
\mathrm{Tr}\left(K_{N}^{(2)}f\right)^{2}
&=&\int_{-1}^{1}\int_{-1}^{1}K_{N}^{(2)}(x,y)f(Ny)K_{N}^{(2)}(y,x)
f(Nx)dx dy\nonumber\\
&=&\frac{1}{N^2}\int_{-N}^{N}\int_{-N}^{N}K_{N}^{(2)}
\left(\frac{x}{N},\frac{y}{N}\right)f(y)K_{N}^{(2)}\left(\frac{y}{N},\frac{x}{N}\right)f(x)dx dy\nonumber\\
&=&\int_{-\infty}^{\infty}\int_{-\infty}^{\infty}K_{\mathrm{sine}}^2(x,y)f(x)f(y)dx dy+O(N^{-1}),\qquad N\rightarrow\infty.\nonumber
\eea
Hence, we find heuristically from (\ref{logdet}) that $\log\det\left(I+K_{N}^{(2)}f\right)$ equals
\bea\label{log1}
&&\int_{-\infty}^{\infty}K_{\mathrm{sine}}(x,x)f(x)dx-\frac{1}{2}\int_{-\infty}^{\infty}\int_{-\infty}^{\infty}K_{\mathrm{sine}}^{2}(x,y)f(x)f(y)dx dy\no\\
&&+\cdots+O(N^{-1}).
\eea

Now we are able to derive the mean and variance of the scaled linear statistics $\sum_{j=1}^{N}F(Nx_j)$.
Taking account of the relation of $f(x)$ and $F(x)$ in (\ref{re}), we have
\be\label{fx}
f(x)=-\lambda F(x)+\frac{\lambda^{2}}{2}F^{2}(x)-\cdots.
\ee
Substituting (\ref{fx}) into (\ref{log1}) gives
\bs
\bea
\log\det\left(I+K_{N}^{(2)}f\right)&=&-\lambda\int_{-\infty}^{\infty}K_{\mathrm{sine}}(x,x)F(x)dx
+\frac{\lambda^{2}}{2}\bigg[\int_{-\infty}^{\infty}K_{\mathrm{sine}}(x,x)F^{2}(x)dx\nonumber\\
&-&\int_{-\infty}^{\infty}\int_{-\infty}^{\infty}K_{\mathrm{sine}}^{2}(x,y)F(x)F(y)dx dy\bigg]+\cdots+O(N^{-1}).\nonumber
\eea
\es
From the coefficients of $\lambda$ and $\lambda^{2}$ and in view of $\log G_{N}^{(2)}(f)=\log\det\big(I+K_{N}^{(2)}f\big)$, we get the following results.
\begin{theorem}
Let $\mu_{N}^{(\mathrm{JUE})}$ and $\mathcal{V}_{N}^{(\mathrm{JUE})}$ be the mean and variance of the scaled linear statistics
$\sum_{j=1}^{N}F(Nx_j)$, respectively. We have as $N\rightarrow\infty$,
\be\label{juem}
\mu_{N}^{(\mathrm{JUE})}=\int_{-\infty}^{\infty}K_{\mathrm{sine}}(x,x)F(x)dx+O(N^{-1}),
\ee
\bea\label{juev}
\mathcal{V}_{N}^{(\mathrm{JUE})}&=&\int_{-\infty}^{\infty}K_{\mathrm{sine}}(x,x)F^{2}(x)dx-\int_{-\infty}^{\infty}\int_{-\infty}^{\infty}K_{\mathrm{sine}}^{2}(x,y)F(x)F(y)dx dy\no\\
&+&O(N^{-1}),
\eea
where $K_{\mathrm{sine}}(x,y)$ is the sine kernel defined by (\ref{sine}).
\end{theorem}
\begin{remark}
The result of the above theorem can also be derived by using the method in the paper \cite{Basor1997}, and it is consistent with the one for the Gaussian unitary ensemble in that paper.
\end{remark}

\subsection{Scaling at the Edge of the Spectrum in JUE}
Contrasting to the previous subsection, we rescale the JUE at the (hard) edge of the spectrum in this subsection. It will be seen that the Bessel kernel arises.
\begin{theorem}\label{jue2}
For $x, y\in\mathbb{R}^{+}$, we have as $N\rightarrow\infty$,
$$
\frac{1}{2N^2}K_{N}^{(2)}\left(1-\frac{x}{2N^2},1-\frac{y}{2N^2}\right)
=K_{\mathrm{Bessel}}^{(a)}(x,y)+O(N^{-1}),
$$
where $K_{\mathrm{Bessel}}^{(a)}(x,y)$ is the Bessel kernel of order $a$ defined by
\be\label{Bessel}
K_{\mathrm{Bessel}}^{(a)}(x,y):=\frac{J_a(\sqrt{x})\sqrt{y}J_a'(\sqrt{y})-J_a'(\sqrt{x})\sqrt{x}J_a(\sqrt{y})}{2(x-y)},
\ee
and $J_a(\cdot)$ is the Bessel function of the first kind of order $a$ \cite[p. 102]{Lebedev}. The error term is uniform for $x$ and $y$ in compact subsets of $\mathbb{R}^{+}$.
\end{theorem}
\begin{proof}
Taking account of (\ref{cd}) and using the Hilb-type asymptotic formula of the Jacobi polynomials \cite[p. 197]{Szego}
\bea\label{asy2}
\left(\sin\frac{\theta}{2}\right)^{a}\left(\cos\frac{\theta}{2}\right)^{b}P_n^{(a,b)}(\cos \theta)&=&\left(n+\frac{a+b+1}{2}\right)^{-a}\frac{\Gamma(n+a+1)}{n!}\left(\frac{\theta}{\sin\theta}\right)^{1/2}\nonumber\\
&\times&J_a\left(\left(n+\frac{a+b+1}{2}\right)\theta\right)+\theta^{1/2}O(n^{-3/2}),
\eea
where $0<\theta<\pi$, we find that $\frac{1}{2N^2}K_{N}^{(2)}\left(1-\frac{x}{2N^2},1-\frac{y}{2N^2}\right)$ equals
\bs
\bea
&&\frac{\Gamma(N+1)\Gamma(N+a+b+1)}{(2N+a+b)\Gamma(N+a)\Gamma(N+b)(x-y)}\Bigg\{\Bigg[J_{a}\left(\frac{N+\frac{a+b-1}{2}}{N}\sqrt{x}\right)J_{a}
\left(\frac{N+\frac{a+b+1}{2}}{N}\sqrt{y}\right)\no\\
&-&J_{a}\left(\frac{N+\frac{a+b+1}{2}}{N}\sqrt{x}\right)J_{a}\left(\frac{N+\frac{a+b-1}{2}}{N}\sqrt{y}\right)\Bigg]+O(N^{-2})\Bigg\},\nonumber
\eea
\es%
uniformly for $x$ and $y$ in compact subsets of $\mathbb{R}^{+}$.
By writing the formula in the square brackets $[\cdots]$ as
\bea
&&J_{a}\left(\frac{N+\frac{a+b-1}{2}}{N}\sqrt{x}\right)\left(J_{a}\left(\frac{N+\frac{a+b+1}{2}}{N}\sqrt{y}\right)-J_{a}\left(\frac{N+\frac{a+b-1}{2}}{N}\sqrt{y}\right)\right)
\nonumber\\
&-&J_{a}\left(\frac{N+\frac{a+b-1}{2}}{N}\sqrt{y}\right)\left(J_{a}\left(\frac{N+\frac{a+b+1}{2}}{N}\sqrt{x}\right)-J_{a}\left(\frac{N+\frac{a+b-1}{2}}{N}\sqrt{x}\right)\right),
\nonumber
\eea
we finally obtain the desired result by taking a large $N$ limit together with the aid of Stirling's formula.
\end{proof}
\begin{remark}
When $x=y$, we have as $N\rightarrow\infty$,
$$
\frac{1}{2N^2}K_{N}^{(2)}\left(1-\frac{x}{2N^2},1-\frac{x}{2N^2}\right)
=K_{\mathrm{Bessel}}^{(a)}(x,x)+O(N^{-1}),
$$
where
$$
K_{\mathrm{Bessel}}^{(a)}(x,x)=\frac{(J_a(\sqrt{x}))^2-J_{a+1}(\sqrt{x})J_{a-1}(\sqrt{x})}{4},
$$
which is obtained by letting $y\rightarrow x$ in (\ref{Bessel}). The error term is uniform for $x$ in compact subsets of $\mathbb{R}^{+}$. The Bessel kernel also arises in the Laguerre unitary ensemble when scaling at the hard edge of the spectrum \cite{Forrester}; see also \cite{Tracy1994B}.
\end{remark}

Similarly as in the previous subsection, we use Theorem \ref{jue2} to compute (\ref{logdet}) term by term as $N\rightarrow\infty$, and replace $f(x)$ by $f(2N^2(1-x))$ in the computations. The first term reads
\bea
\mathrm{Tr}K_{N}^{(2)}f
&=&\int_{-1}^{1}K_{N}^{(2)}(x,x)f(2N^2(1-x))dx\nonumber\\
&=&\int_{0}^{4N^2}\frac{1}{2N^2}K_{N}^{(2)}
\left(1-\frac{x}{2N^2},1-\frac{x}{2N^2}\right)f(x)dx\nonumber\\
&=&\int_{0}^{\infty}K_{\mathrm{Bessel}}^{(a)}(x,x)f(x)dx+O(N^{-1}),\qquad N\rightarrow\infty.\nonumber
\eea
\begin{remark}
We assume that $f(\cdot)$ is a continuous real-valued function belonging to $L^{1}(\mathbb{R}^{+})$ and vanishes at $+\infty$.
\end{remark}
The second term gives
\bea
\mathrm{Tr}\left(K_{N}^{(2)}f\right)^{2}
&=&\int_{-1}^{1}\int_{-1}^{1}K_{N}^{(2)}(x,y)f(2N^2(1-y))K_{N}^{(2)}(y,x)
f(2N^2(1-x))dx dy\nonumber\\
&=&\int_{0}^{\infty}\int_{0}^{\infty}\left(K_{\mathrm{Bessel}}^{(a)}(x,y)\right)^2f(x)f(y)dx dy+O(N^{-1}),\quad N\rightarrow\infty.\nonumber
\eea
It follows, again heuristically, from (\ref{logdet}) that $\log\det\left(I+K_{N}^{(2)}f\right)$ equals
\bea
&&\int_{0}^{\infty}K_{\mathrm{Bessel}}^{(a)}(x,x)f(x)dx-\frac{1}{2}\int_{0}^{\infty}\int_{0}^{\infty}\left(K_{\mathrm{Bessel}}^{(a)}(x,y)\right)^2f(x)f(y)dx dy\no\\
&&+\cdots+O(N^{-1}).\no
\eea
Substituting (\ref{fx}) into the above gives
\bea
&&\log\det\left(I+K_{N}^{(2)}f\right)=-\lambda\int_{0}^{\infty}K_{\mathrm{Bessel}}^{(a)}(x,x)F(x)dx\no\\
&&+\frac{\lambda^{2}}{2}
\bigg[\int_{0}^{\infty}K_{\mathrm{Bessel}}^{(a)}(x,x)F^{2}(x)dx
-\int_{0}^{\infty}\int_{0}^{\infty}\left(K_{\mathrm{Bessel}}^{(a)}(x,y)\right)^2F(x)F(y)dx dy\bigg]\no\\
&&+\cdots+O(N^{-1}).\nonumber
\eea
Then, the following theorem follows.
\begin{theorem}
Let $\tilde{\mu}_{N}^{(\mathrm{JUE})}$ and $\tilde{\mathcal{V}}_{N}^{(\mathrm{JUE})}$ be the mean and variance of the scaled linear statistics
$\sum_{j=1}^{N}F(2N^2(1-x_j))$, respectively. Then as $N\rightarrow\infty$,
\be\label{juem2}
\tilde{\mu}_{N}^{(\mathrm{JUE})}=\int_{0}^{\infty}K_{\mathrm{Bessel}}^{(a)}(x,x)F(x)dx+O(N^{-1}),
\ee
\bea\label{juev2}
\tilde{\mathcal{V}}_{N}^{(\mathrm{JUE})}&=&\int_{0}^{\infty}K_{\mathrm{Bessel}}^{(a)}(x,x)F^{2}(x)dx-\int_{0}^{\infty}\int_{0}^{\infty}
\left(K_{\mathrm{Bessel}}^{(a)}(x,y)\right)^2F(x)F(y)dx dy\no\\
&+&O(N^{-1}),
\eea
where $K_{\mathrm{Bessel}}^{(a)}(x,y)$ is the Bessel kernel defined by (\ref{Bessel}).
\end{theorem}
\begin{remark}
The result of the above theorem is consistent with the one for the Laguerre unitary ensemble by scaling at the hard edge of the spectrum \cite{Basor1997}.
\end{remark}

\section{Jacobi Symplectic Ensemble (JSE)}
\subsection{Finite $N$ Case for the MGF in JSE}
In this case, $w(x)=(1-x)^a(1+x)^b,\;\;x\in [-1,1],\;a, b>0$.
The authors \cite{Min201601} expressed $G_{N}^{(4)}(f)$ as a Fredholm determinant based on the work of Dieng and Tracy \cite{Dieng} and Tracy and Widom \cite{Tracy1998}.
Define
$$
\psi_{j}^{(4)}(x):=\pi_{j}(x)\sqrt{w(x)},\;\;j=0,1,2,\ldots,
$$
where $\pi_{j}(x)$ is an arbitrary polynomial of degree $j$, and
$$
M^{(4)}:=\left[\int_{-1}^{1}\left(\psi_{j}^{(4)}(x)\frac{d}{dx}\psi_{k}^{(4)}(x)-\psi_{k}^{(4)}(x)\frac{d}{dx}\psi_{j}^{(4)}(x)\right)dx\right]_{j,k=0}^{2N-1}
$$
with its inverse denoted by
$$
\left(M^{(4)}\right)^{-1}=:(\mu_{jk})_{j,k=0}^{2N-1}.
$$
It was shown in \cite{Min201601} that
\be\label{gn4}
\left[G_{N}^{(4)}(f)\right]^{2}=\det\left(I+2K_{N}^{(4)}f-K_{N}^{(4)}\varepsilon f'\right),
\ee
where $K_{N}^{(4)}$ and $\varepsilon$ are integral operators with kernel
\be\label{kn4}
K_{N}^{(4)}(x,y):=-\sum_{j,k=0}^{2N-1}\mu_{jk}\psi_{j}^{(4)}(x)\frac{d}{dy}\psi_{k}^{(4)}(y)
\ee
and
$$
\varepsilon(x,y):=\frac{1}{2}\mathrm{sgn}(x-y),
$$
respectively. We require that $f\in C^{1}[-1,1]$ and vanishes at the endpoints $\pm 1$.
\begin{remark}
If $g(x)$ is an integrable function on $[-1,1]$, then
$$
\varepsilon g(x)=\int_{-1}^{1}\varepsilon(x,y) g(y)dy=\frac{1}{2}\left(\int_{-1}^{x}g(y)dy-\int_{x}^{1}g(y)dy\right),\qquad x\in [-1,1].
$$
In addition, it is easy to see that $\varepsilon(y,x)=-\varepsilon(x,y)$, i.e., $\varepsilon^{t}=-\varepsilon$, where $t$ denotes the transpose.
\end{remark}

The fundamental theorem of calculus implies the following result \cite{Dieng}; see also \cite{Min201601}.
\begin{lemma}\label{lem}
Let $D$ be the operator that acts by differentiation. Then for any function $g\in C^{1}[-1,1]$ and $g(-1)=g(1)=0$, we have $D\varepsilon g(x)=\varepsilon D g(x)=g(x)$, i.e.,
$D\varepsilon=\varepsilon D=I$.
\end{lemma}

Similarly to the discussions in \cite{Dieng,Min201601,Tracy1998},
we choose a special $\psi_{j}^{(4)}$ to simplify  $M^{(4)}$ as much as possible. To proceed, let
$$
\psi_{2j+1}^{(4)}(x):=\frac{1}{\sqrt{2}}(1-x^2)\varphi_{2j+1}^{(4)}(x),\quad \psi_{2j}^{(4)}(x):=-\frac{1}{\sqrt{2}}\varepsilon\varphi_{2j+1}^{(4)}(x),\;\;j=0,1,2,\ldots,
$$
where $\varphi_{j}^{(4)}(x)$ is given by
$$
\varphi_{j}^{(4)}(x)=\frac{P_{j}^{(a-1,b-1)}(x)}{\sqrt{h_{j}^{(a-1,b-1)}}}(1-x)^{\frac{a}{2}-1}(1+x)^{\frac{b}{2}-1},
$$
and $P_{j}^{(a-1,b-1)}(x),\; j=0,1,\ldots$ are the usual Jacobi polynomials with the orthogonality condition
$$
\int_{-1}^{1}P_{j}^{(a-1,b-1)}(x)P_{k}^{(a-1,b-1)}(x)(1-x)^{a-1}(1+x)^{b-1} dx=h_{j}^{(a-1,b-1)}\delta_{jk}.
$$
It can be shown that $\psi_{j}^{(4)}(x)$ is equal to $(1-x)^{a/2}(1+x)^{b/2}$ multiplied by a polynomial of degree $j$. Similarly as the Laguerre symplectic ensemble case studied in \cite[Theorem 3.10]{Min201601}, $M^{(4)}$ is computed to be the direct sum of $N$ copies of
\bs$
\begin{pmatrix}
0&1\\
-1&0
\end{pmatrix}
$\es
by using the orthogonality, namely
$$
M^{(4)}=\begin{pmatrix}
0&1&0&0&\cdots&0&0\\
-1&0&0&0&\cdots&0&0\\
0&0&0&1&\cdots&0&0\\
0&0&-1&0&\cdots&0&0\\
\vdots&\vdots&\vdots&\vdots&&\vdots&\vdots\\
0&0&0&0&\cdots&0&1\\
0&0&0&0&\cdots&-1&0
\end{pmatrix}_{2N\times 2N}.
$$
It follows that $\left(M^{(4)}\right)^{-1}=-M^{(4)}$, so $\mu_{2j,2j+1}=-1,\; \mu_{2j+1,2j}=1,\; j=0,1,\ldots,N-1$, and $\mu_{jk}=0$ for other cases.

\begin{lemma}
We have
\be\label{kn4a}
K_N^{(4)}(x,y)=\frac{1}{2}S_N^{(4)}(x,y)+\frac{1}{2}C_{2N}^{(4)}\varepsilon\varphi_{2N+1}^{(4)}(x)
\varphi_{2N}^{(4)}(y),
\ee
where
$$
C_{2N}^{(4)}=\sqrt{\frac{(2N+1)(2N+a)(2N+b)(2N+a+b-1)}{(4N+a+b+1)(4N+a+b-1)}}
$$
and
$$
S_N^{(4)}(x,y)=\sum_{j=0}^{2N}(1-x^2)\varphi_{j}^{(4)}(x)\varphi_{j}^{(4)}(y).
$$
\end{lemma}
\begin{proof}
From (\ref{kn4}) we find that $K_{N}^{(4)}(x,y)$ equals
\bea
&&\sum_{j=0}^{N-1}\psi_{2j}(x)\psi_{2j+1}'(y)-\sum_{j=0}^{N-1}\psi_{2j+1}(x)\psi_{2j}'(y)\nonumber\\
&&=\frac{1}{2}\sum_{j=0}^{N-1}(1-x^2)\varphi_{2j+1}^{(4)}(x)\varphi_{2j+1}^{(4)}(y)-\frac{1}{2}\sum_{j=0}^{N-1}\varepsilon\varphi_{2j+1}^{(4)}(x)
\left[(1-y^2)\varphi_{2j+1}^{(4)}(y)\right]'.\nonumber
\eea
By using the recurrence formulas for the Jacobi polynomials \cite[Sec. 4.5]{Szego}
\bea\label{rf1}
&&(2n+a+b)(1-x^2)\frac{d}{dx}P_n^{(a,b)}(x)\no\\
&=&n\left[a-b-(2n+a+b)x\right]P_n^{(a,b)}(x)+2(n+a)(n+b)P_{n-1}^{(a,b)}(x)
\eea
and
\bea\label{rf2}
&&(2n+a+b+1)[(2n+a+b)(2n+a+b+2)x+a^2-b^2]P_n^{(a,b)}(x)\nonumber\\
&=&2(n+1)(n+a+b+1)(2n+a+b)P_{n+1}^{(a,b)}(x)\no\\
&+&2(n+a)(n+b)(2n+a+b+2)P_{n-1}^{(a,b)}(x),
\eea
we obtain
$$
\left[(1-y^2)\varphi_{2j+1}^{(4)}(y)\right]'=C_{2j}^{(4)}\varphi_{2j}^{(4)}(y)-C_{2j+1}^{(4)}\varphi_{2j+2}^{(4)}(y),
$$
where
$$
C_j^{(4)}:=\sqrt{\frac{(j+1)(j+a)(j+b)(j+a+b-1)}{(2j+a+b+1)(2j+a+b-1)}}.
$$
It follows that
\bea\label{com1}
K_{N}^{(4)}(x,y)&=&\frac{1}{2}\sum_{j=0}^{N-1}(1-x^2)\varphi_{2j+1}^{(4)}(x)\varphi_{2j+1}^{(4)}(y)\no\\
&+&\frac{1}{2}\sum_{j=0}^{N}\left[C_{2j-1}^{(4)}\varepsilon\varphi_{2j-1}^{(4)}(x)-C_{2j}^{(4)}\varepsilon\varphi_{2j+1}^{(4)}(x)\right]\varphi_{2j}^{(4)}(y)\no\\
&+&\frac{1}{2}C_{2N}^{(4)}\varepsilon\varphi_{2N+1}^{(4)}(x)\varphi_{2N}^{(4)}(y).
\eea
Using (\ref{rf1}) and (\ref{rf2}) again, we find
$$
\left[(1-x^2)\varphi_{2j}^{(4)}(x)\right]'=C_{2j-1}^{(4)}\varphi_{2j-1}^{(4)}(x)-C_{2j}^{(4)}\varphi_{2j+1}^{(4)}(x).
$$
Then from Lemma \ref{lem} we have
\be\label{com2}
(1-x^2)\varphi_{2j}^{(4)}(x)=\varepsilon\left[(1-x^2)\varphi_{2j}^{(4)}(x)\right]'=C_{2j-1}^{(4)}\varepsilon\varphi_{2j-1}^{(4)}(x)-C_{2j}^{(4)}\varepsilon\varphi_{2j+1}^{(4)}(x).
\ee
The combination of (\ref{com1}) and (\ref{com2}) gives
\bea
K_{N}^{(4)}(x,y)&=&\frac{1}{2}\sum_{j=0}^{N-1}(1-x^2)\varphi_{2j+1}^{(4)}(x)\varphi_{2j+1}^{(4)}(y)
+\frac{1}{2}\sum_{j=0}^{N}(1-x^2)\varphi_{2j}^{(4)}(x)\varphi_{2j}^{(4)}(y)\no\\
&+&\frac{1}{2}C_{2N}^{(4)}\varepsilon\varphi_{2N+1}^{(4)}(x)\varphi_{2N}^{(4)}(y)\no\\
&=&\frac{1}{2}\sum_{j=0}^{2N}(1-x^2)\varphi_{j}^{(4)}(x)\varphi_{j}^{(4)}(y)+\frac{1}{2}C_{2N}^{(4)}\varepsilon\varphi_{2N+1}^{(4)}(x)
\varphi_{2N}^{(4)}(y).\no
\eea
The proof is complete.
\end{proof}

\begin{theorem}
For the Jacobi symplectic ensemble, we have
\be\label{gn4d}
\left[G_{N}^{(4)}(f)\right]^{2}=\det(I+T_{\mathrm{JSE}}),
\ee
where
$$
T_{\mathrm{JSE}}:=S_{N}^{(4)}f-\frac{1}{2}S_{N}^{(4)}\varepsilon f'+C_{2N}^{(4)}\left(\varepsilon\varphi_{2N+1}^{(4)}\right)\otimes\varphi_{2N}^{(4)}f
+\frac{1}{2}C_{2N}^{(4)}\left(\varepsilon\varphi_{2N+1}^{(4)}\right)\otimes\left(\varepsilon\varphi_{2N}^{(4)}\right) f'.
$$
\end{theorem}
\begin{proof}
Substituting $K_N^{(4)}$ with the kernel given by (\ref{kn4a}) into (\ref{gn4}), we obtain the desired result with the aid of the property $(u\otimes v)A=u\otimes(A^{t}v)$ for integral operators.
\end{proof}

Finally we mention the following expansion formula
\be\label{logdet4}
\log\det(I+T_{\mathrm{JSE}})=\mathrm{Tr}\log(I+T_{\mathrm{JSE}})=\mathrm{Tr}\:T_{\mathrm{JSE}}-\frac{1}{2}\mathrm{Tr}\:T_{\mathrm{JSE}}^{2}
+\frac{1}{3}\mathrm{Tr}\:T_{\mathrm{JSE}}^{3}-\cdots,
\ee
which will be used in the asymptotic analysis in the next subsections.

\subsection{Scaling in the Bulk of the Spectrum in JSE}
Similarly as Theorem \ref{jue1}, we have the following theorem.
\begin{theorem}\label{jse1}
For $x, y\in\mathbb{R}$, we have as $N\rightarrow\infty$,
$$
\frac{1}{2N}S_{N}^{(4)}\left(\frac{x}{2N},\frac{y}{2N}\right)
=K_{\mathrm{sine}}(x,y)+O(N^{-1}),
$$
where $K_{\mathrm{sine}}(x,y)$ is the sine kernel given by (\ref{sine}). The error term is uniform for $x$ and $y$ in compact subsets of $\mathbb{R}$.
\end{theorem}

\begin{theorem}\label{jse2}
For $x\in\mathbb{R}$, we have as $N\rightarrow\infty$,
$$
\varphi_{2N}^{(4)}\left(\frac{x}{2N}\right)
=\sqrt{\frac{2}{\pi}}\sin\left[\frac{1}{4}\left(\pi+2\pi a-2(4N-1+a+b)\arccos\frac{x}{2N}\right)\right]+O(N^{-1}),
$$
$$
\varphi_{2N+1}^{(4)}\left(\frac{x}{2N}\right)
=\sqrt{\frac{2}{\pi}}\sin\left[\frac{1}{4}\left(\pi+2\pi a-2(4N+1+a+b)\arccos\frac{x}{2N}\right)\right]+O(N^{-1}),
$$
\bea
\varepsilon\varphi_{2N}^{(4)}\left(\frac{x}{2N}\right)
&=&\frac{1}{2N\sqrt{2\pi}}\bigg\{\sin\left[\frac{1}{4}\left(\pi+2\pi a-2(4N+1+a+b)\arccos\frac{x}{2N}\right)\right]\nonumber\\
&-&\sin\left[\frac{1}{4}\left(\pi+2\pi a-2(4N-3+a+b)\arccos\frac{x}{2N}\right)\right]\bigg\}+O(N^{-2}),\nonumber
\eea
\bea
\varepsilon\varphi_{2N+1}^{(4)}\left(\frac{x}{2N}\right)
&=&\frac{1}{2N\sqrt{2\pi}}\bigg\{\sin\left[\frac{1}{4}\left(\pi+2\pi a-2(4N+3+a+b)\arccos\frac{x}{2N}\right)\right]\nonumber\\
&-&\sin\left[\frac{1}{4}\left(\pi+2\pi a-2(4N-1+a+b)\arccos\frac{x}{2N}\right)\right]\bigg\}+O(N^{-2}).\nonumber
\eea
The error terms are uniform for $x$ in compact subsets of $\mathbb{R}$.
\end{theorem}
\begin{proof}
By using the asymptotic formula of the Jacobi polynomials (\ref{asy}), we obtain the desired results after direct calculations.
\end{proof}
\begin{remark}
It is easy to see from the above theorem that as $N\rightarrow\infty$,
$$
\varphi_{2N}^{(4)}\left(\frac{x}{2N}\right)=O(1),\qquad\qquad\qquad\qquad \varphi_{2N+1}^{(4)}\left(\frac{x}{2N}\right)=O(1),
$$
$$
\varepsilon\varphi_{2N}^{(4)}\left(\frac{x}{2N}\right)=O(N^{-1}),\qquad\qquad\qquad\qquad \varepsilon\varphi_{2N+1}^{(4)}\left(\frac{x}{2N}\right)=O(N^{-1}),
$$
uniformly for $x$ in compact subsets of $\mathbb{R}$.
\end{remark}

We now use Theorem \ref{jse1} and \ref{jse2} to compute (\ref{logdet4}) as $N\rightarrow\infty$, and change $f(x)$ to $f(2Nx)$ in the calculations (in this case $f'(x)$ becomes $2Nf'(2Nx)$). We first consider $\mathrm{Tr}\:T_{\mathrm{JSE}}$:
\bea\label{trjse}
\mathrm{Tr}\:T_{\mathrm{JSE}}&=&\mathrm{Tr}\:S_{N}^{(4)}f-\mathrm{Tr}\:\frac{1}{2}S_{N}^{(4)}\varepsilon f'+\mathrm{Tr}\:C_{2N}^{(4)}\left(\varepsilon\varphi_{2N+1}^{(4)}\right)\otimes\varphi_{2N}^{(4)}f\no\\
&+&\mathrm{Tr}\:\frac{1}{2}C_{2N}^{(4)}\left(\varepsilon\varphi_{2N+1}^{(4)}\right)\otimes\left(\varepsilon\varphi_{2N}^{(4)}\right) f'.
\eea
The first term gives
\bea
\mathrm{Tr}\:S_{N}^{(4)}f
&=&\int_{-1}^{1}S_{N}^{(4)}(x,x)f(2Nx)dx\nonumber\\
&=&\int_{-2N}^{2N}\frac{1}{2N}S_{N}^{(4)}
\left(\frac{x}{2N},\frac{x}{2N}\right)f(x)dx\nonumber\\
&=&\int_{-\infty}^{\infty}K_{\mathrm{sine}}(x,x)f(x)dx+O(N^{-1}),\qquad N\rightarrow\infty.\nonumber
\eea
\begin{remark}
We assume that $f(\cdot)$ is smooth and sufficiently decreasing at $\pm\infty$ to make the integrals well-defined.
\end{remark}
It can be shown that the rest terms have contributions of $O(N^{-1})$, where we have used the fact
$$
\int_{x}^{\infty}K_{\mathrm{sine}}(x,y)dy-\int_{-\infty}^{x}K_{\mathrm{sine}}(x,y)dy=0
$$
in the calculation of $\mathrm{Tr}\:\frac{1}{2}S_{N}^{(4)}\varepsilon f'$. Hence,
\be\label{tr1}
\mathrm{Tr}\:T_{\mathrm{JSE}}=\int_{-\infty}^{\infty}K_{\mathrm{sine}}(x,x)f(x)dx+O(N^{-1}).
\ee
Next, we consider $\mathrm{Tr}\:T_{\mathrm{JSE}}^{2}$:
\bs
\bea\label{trjse2}
\mathrm{Tr}\:T_{\mathrm{JSE}}^{2}
&=&\mathrm{Tr}\left(S_{N}^{(4)}f\right)^2-\mathrm{Tr}\:S_{N}^{(4)}f S_{N}^{(4)}\varepsilon f'+\mathrm{Tr}\:2C_{2N}^{(4)}S_{N}^{(4)}f\left(\varepsilon\varphi_{2N+1}^{(4)}
\otimes\varphi_{2N}^{(4)}\right)f\no\\
&+&\mathrm{Tr}\:C_{2N}^{(4)}S_{N}^{(4)}f\left(\varepsilon\varphi_{2N+1}^{(4)}
\otimes\varepsilon\varphi_{2N}^{(4)}\right)f'+\mathrm{Tr}\:\frac{1}{4}\left(S_{N}^{(4)}\varepsilon f'\right)^2\no\\
&-&\mathrm{Tr}\:C_{2N}^{(4)}S_{N}^{(4)}\varepsilon f'\left(\varepsilon\varphi_{2N+1}^{(4)}\otimes\varphi_{2N}^{(4)}\right)f-\mathrm{Tr}\:\frac{1}{2}C_{2N}^{(4)}S_{N}^{(4)}\varepsilon f'\left(\varepsilon\varphi_{2N+1}^{(4)}\otimes\varepsilon\varphi_{2N}^{(4)}\right)f'
\no\\
&+&\mathrm{Tr}\left(C_{2N}^{(4)}\right)^2\left(\varepsilon\varphi_{2N+1}^{(4)}
\otimes\varphi_{2N}^{(4)}f\right)^2+\mathrm{Tr}\:\frac{1}{4}\left(C_{2N}^{(4)}\right)^2\left(\varepsilon\varphi_{2N+1}^{(4)}
\otimes\varepsilon\varphi_{2N}^{(4)}f'\right)^2\no\\
&+&\mathrm{Tr}\left(C_{2N}^{(4)}\right)^2\left(\varepsilon\varphi_{2N+1}^{(4)}
\otimes\varphi_{2N}^{(4)}f\right)\left(\varepsilon\varphi_{2N+1}^{(4)}
\otimes\varepsilon\varphi_{2N}^{(4)}f'\right).
\eea
\es
We find as $N\rightarrow\infty$,
$$
\mathrm{Tr}\left(S_{N}^{(4)}f\right)^2=\int_{-\infty}^{\infty}\int_{-\infty}^{\infty}K_{\mathrm{sine}}^2(x,y)f(x)f(y)dxdy+O(N^{-1}),
$$
$$
\mathrm{Tr}\:S_{N}^{(4)}f S_{N}^{(4)}\varepsilon f'=-\frac{1}{2\pi^2}\int_{-\infty}^{\infty}\int_{-\infty}^{\infty}\mathrm{Si}^2(x-y)f'(x)f'(y)dxdy+O(N^{-1}),
$$
$$
\mathrm{Tr}\:\frac{1}{4}\left(S_{N}^{(4)}\varepsilon f'\right)^2=-\frac{1}{4\pi^2}\int_{-\infty}^{\infty}\int_{-\infty}^{\infty}\mathrm{Si}^2(x-y)f'(x)f'(y)dxdy+O(N^{-1}),
$$
where we have used integration by parts to obtain the second equality and the formula
$$
\int_{-\infty}^{x}K_{\mathrm{sine}}(y,z)dz-\int_{x}^{\infty}K_{\mathrm{sine}}(y,z)dz=\frac{2}{\pi}\:\mathrm{Si}(x-y),
$$
and $\mathrm{Si}(x)$ is the sine integral defined by \cite[Sec. 3.3]{Lebedev}
$$
\mathrm{Si}(x):=\int_{0}^{x}\frac{\sin t}{t}dt.
$$
The rest terms in (\ref{trjse2}) are proven to have contributions of $O(N^{-1})$ after some elaborate computations. Hence,
\bea\label{tr2}
\mathrm{Tr}\:T_{\mathrm{JSE}}^{2}&=&\int_{-\infty}^{\infty}\int_{-\infty}^{\infty}K_{\mathrm{sine}}^2(x,y)f(x)f(y)dxdy\no\\
&+&\frac{1}{4\pi^2}\int_{-\infty}^{\infty}\int_{-\infty}^{\infty}\mathrm{Si}^2(x-y)f'(x)f'(y)dxdy+O(N^{-1}).
\eea

Proceeding as in the JUE case, we substitute (\ref{fx}) into (\ref{tr1}) and (\ref{tr2}) and finally heuristically obtain from (\ref{logdet4}) that
\bea
\log\det(I+T_{\mathrm{JSE}})
&=&-\lambda\int_{-\infty}^{\infty}K_{\mathrm{sine}}(x,x)F(x)dx+\frac{\lambda^2}{2}\bigg[\int_{-\infty}^{\infty}K_{\mathrm{sine}}(x,x)F^{2}(x)dx\no\\
&-&\int_{-\infty}^{\infty}\int_{-\infty}^{\infty}K_{\mathrm{sine}}^{2}(x,y)F(x)F(y)dxdy\no\\
&-&\frac{1}{4\pi^2}\int_{-\infty}^{\infty}\int_{-\infty}^{\infty}\mathrm{Si}^2(x-y)F'(x)F'(y)dxdy\bigg]+\cdots+O(N^{-1}).\no
\eea
Since we have $\log G_{N}^{(4)}(f)=\frac{1}{2}\log\det(I+T_{\mathrm{JSE}})$ from (\ref{gn4d}), the following theorem follows.
\begin{theorem}
Let $\mu_{N}^{(\mathrm{JSE})}$ and $\mathcal{V}_{N}^{(\mathrm{JSE})}$ be the mean and variance of the scaled linear statistics
$\sum_{j=1}^{N}F(2Nx_j)$. We have as $N\rightarrow\infty$,
$$
\mu_{N}^{(\mathrm{JSE})}=\frac{1}{2}\mu_{N}^{(\mathrm{JUE})}+O(N^{-1}),
$$
$$
\mathcal{V}_{N}^{(\mathrm{JSE})}
=\frac{1}{2}\mathcal{V}_{N}^{(\mathrm{JUE})}-\frac{1}{8\pi^2}\int_{-\infty}^{\infty}\int_{-\infty}^{\infty}\mathrm{Si}^2(x-y)F'(x)F'(y)dxdy+O(N^{-1}),
$$
where $\mu_{N}^{(\mathrm{JUE})}$ and $\mathcal{V}_{N}^{(\mathrm{JUE})}$ are given by (\ref{juem}) and (\ref{juev}), respectively.
\end{theorem}

\subsection{Scaling at the Edge of the Spectrum in JSE}
Similarly as Theorem \ref{jue2}, we have the following result.
\begin{theorem}\label{jse3}
For $x, y\in\mathbb{R}^{+}$, we have as $N\rightarrow\infty$,
$$
\frac{1}{8N^2}S_{N}^{(4)}\left(1-\frac{x}{8N^2},1-\frac{y}{8N^2}\right)
=\sqrt{\frac{x}{y}}\:K_{\mathrm{Bessel}}^{(a-1)}(x,y)+O(N^{-1}),
$$
where $K_{\mathrm{Bessel}}^{(a-1)}(x,y)$ is the Bessel kernel of order $a-1$ given by
$$
K_{\mathrm{Bessel}}^{(a-1)}(x,y)=\frac{J_{a-1}(\sqrt{x})\sqrt{y}J_{a-1}'(\sqrt{y})-J_{a-1}'(\sqrt{x})\sqrt{x}J_{a-1}(\sqrt{y})}{2(x-y)}.
$$
The error term is uniform for $x$ and $y$ in compact subsets of $\mathbb{R}^{+}$.
\end{theorem}

\begin{theorem}\label{jse4}
For $x\in\mathbb{R}^{+}$, we have as $N\rightarrow\infty$,
$$
\varphi_{2N}^{(4)}\left(1-\frac{x}{8N^2}\right)
=(2N)^{3/2}\:\frac{J_{a-1}(\sqrt{x})}{\sqrt{x}}+O(N^{1/2}),
$$
$$
\varphi_{2N+1}^{(4)}\left(1-\frac{x}{8N^2}\right)
=(2N)^{3/2}\:\frac{J_{a-1}(\sqrt{x})}{\sqrt{x}}+O(N^{1/2}),
$$
$$
\varepsilon\varphi_{2N}^{(4)}\left(1-\frac{x}{8N^2}\right)
=2^{-3/2}N^{-1/2}\left(1-2\mathbf{J}_{a-1}(\sqrt{x})\right)+O(N^{-3/2}),
$$
$$
\varepsilon\varphi_{2N+1}^{(4)}\left(1-\frac{x}{8N^2}\right)
=2^{-3/2}N^{-1/2}\left(1-2\mathbf{J}_{a-1}(\sqrt{x})\right)+O(N^{-3/2}),
$$
where
\be\label{jx}
\mathbf{J}_{a-1}(x):=\int_{0}^{x}J_{a-1}(t)dt.
\ee
The error terms are uniform for $x$ in compact subsets of $\mathbb{R}^{+}$.
\end{theorem}
\begin{proof}
The results come from direct computations by using the large $n$ Hilb-type asymptotic formula of the Jacobi polynomials (\ref{asy2}), and the formula
$$
\int_{x}^{\infty}\frac{J_{a-1}(\sqrt{y})}{\sqrt{y}}dy-\int_{0}^{x}\frac{J_{a-1}(\sqrt{y})}{\sqrt{y}}dy=2\left(1-2\mathbf{J}_{a-1}(\sqrt{x})\right),
$$
where use has been made of the fact that $\int_{0}^{\infty}J_{a-1}(t)dt=1$ (see, e.g., \cite[p. 659]{GR}).
\end{proof}

Using Theorem \ref{jse3} and \ref{jse4} to compute (\ref{trjse}) term by term and changing $f(x)$ to $f(8N^2(1-x))$, we find
$$
\mathrm{Tr}\:S_{N}^{(4)}f
=\int_{0}^{\infty}K_{\mathrm{Bessel}}^{(a-1)}(x,x)f(x)dx+O(N^{-1}),
$$
$$
\mathrm{Tr}\:\frac{1}{2}S_{N}^{(4)}\varepsilon f'=-\frac{1}{4}\int_{0}^{\infty}L^{(a-1)}(x,x)f'(x)dx+O(N^{-1}),
$$
$$
\mathrm{Tr}\:C_{2N}^{(4)}\left(\varepsilon\varphi_{2N+1}^{(4)}\right)\otimes\varphi_{2N}^{(4)}f=\frac{1}{8}\int_{0}^{\infty}
\frac{J_{a-1}(\sqrt{x})}{\sqrt{x}}\left(1-2\mathbf{J}_{a-1}(\sqrt{x})\right)f(x)dx+O(N^{-1}),
$$
$$
\mathrm{Tr}\:\frac{1}{2}C_{2N}^{(4)}\left(\varepsilon\varphi_{2N+1}^{(4)}\right)\otimes\left(\varepsilon\varphi_{2N}^{(4)}\right) f'=-\frac{1}{16}\int_{0}^{\infty}\left(1-2\mathbf{J}_{a-1}(\sqrt{x})\right)^2f'(x)dx+O(N^{-1}),
$$
where
\be\label{lxy}
L^{(a-1)}(x,y):=\int_{0}^{x}\sqrt{\frac{y}{z}}\:K_{\mathrm{Bessel}}^{(a-1)}(y,z)dz-\int_{x}^{\infty}\sqrt{\frac{y}{z}}\:K_{\mathrm{Bessel}}^{(a-1)}(y,z)dz.
\ee
\begin{remark}
We assume that $f(\cdot)$ is smooth and sufficiently decreasing at infinity to make the integrals well-defined.
\end{remark}
\noindent Through integration by parts, we have the formula
$$
\int_{0}^{\infty}\left(1-2\mathbf{J}_{a-1}(\sqrt{x})\right)^2f'(x)dx=2\int_{0}^{\infty}
\frac{J_{a-1}(\sqrt{x})}{\sqrt{x}}\left(1-2\mathbf{J}_{a-1}(\sqrt{x})\right)f(x)dx.
$$
It follows that
\bea\label{tr11}
\mathrm{Tr}\:T_{\mathrm{JSE}}&=&\int_{0}^{\infty}K_{\mathrm{Bessel}}^{(a-1)}(x,x)f(x)dx\no\\
&+&\frac{1}{4}\int_{0}^{\infty}L^{(a-1)}(x,x)f'(x)dx+O(N^{-1}).
\eea
Similarly, we obtain from (\ref{trjse2}) by changing $f(x)$ to $f(8N^2(1-x))$ that $\mathrm{Tr}\:T_{\mathrm{JSE}}^{2}$ equals
\bs
\bea\label{tr22}
&&\int_{0}^{\infty}\int_{0}^{\infty}\big(K_{\mathrm{Bessel}}^{(a-1)}(x,y)\big)^2f(x)f(y)dxdy\no\\
&&+\frac{1}{2}\int_{0}^{\infty}\int_{0}^{\infty}\sqrt{\frac{x}{y}}\:K_{\mathrm{Bessel}}^{(a-1)}(x,y)L^{(a-1)}(x,y)f'(x)f(y)dxdy\nonumber\\
&&+\frac{1}{4}\int_{0}^{\infty}\int_{0}^{\infty}\frac{J_{a-1}(\sqrt{x})}{\sqrt{y}}K_{\mathrm{Bessel}}^{(a-1)}(x,y)\left(1-2\mathbf{J}_{a-1}(\sqrt{x})\right)
f(x)f(y)dxdy\nonumber\\
&&+\frac{1}{16}\int_{0}^{\infty}\int_{0}^{\infty}L^{(a-1)}(x,y)L^{(a-1)}(y,x)f'(x)f'(y)dxdy\nonumber\\
&&-\frac{1}{16}\int_{0}^{\infty}\int_{0}^{\infty}\frac{J_{a-1}(\sqrt{x})}{\sqrt{x}}\left(1-2\mathbf{J}_{a-1}(\sqrt{y})\right)
\left(L^{(a-1)}(x,y)-L^{(a-1)}(y,x)\right)f(x)f'(y)dxdy\nonumber\\
&&+\frac{1}{32}\int_{0}^{\infty}\int_{0}^{\infty}\left(1-2\mathbf{J}_{a-1}(\sqrt{x})\right)\left(1-2\mathbf{J}_{a-1}(\sqrt{y})\right)L^{(a-1)}(x,y)f'(x)f'(y)dxdy\no\\
&&+O(N^{-1}),
\eea
\es%
where we have used integration by parts to simplify the results.

Similarly as in Section 3.2, by substituting (\ref{fx}) into (\ref{tr11}) and (\ref{tr22}) and using the fact that $\log G_{N}^{(4)}(f)=\frac{1}{2}\log\det(I+T_{\mathrm{JSE}})$, we heuristically obtain the following theorem.
\begin{theorem}
Denoting by $\tilde{\mu}_{N}^{(\mathrm{JSE})}$ and $\tilde{\mathcal{V}}_{N}^{(\mathrm{JSE})}$ the mean and variance of the linear statistics
$\sum_{j=1}^{N}F(8N^2(1-x_j))$, we have as $N\rightarrow\infty$,
$$
\tilde{\mu}_{N}^{(\mathrm{JSE})}=\frac{1}{2}\tilde{\mu}_{N}^{(\mathrm{JUE},a-1)}+\frac{1}{8}\int_{0}^{\infty}L^{(a-1)}(x,x)F'(x)dx+O(N^{-1}),
$$
\bs
\bea
\tilde{\mathcal{V}}_{N}^{(\mathrm{JSE})}
&=&\frac{1}{2}\tilde{\mathcal{V}}_{N}^{(\mathrm{JUE},a-1)}-\frac{1}{4}\int_{0}^{\infty}\int_{0}^{\infty}\sqrt{\frac{x}{y}}\:K_{\mathrm{Bessel}}^{(a-1)}(x,y)L^{(a-1)}(x,y)
F'(x)F(y)dxdy\nonumber\\
&-&\frac{1}{8}\int_{0}^{\infty}\int_{0}^{\infty}\frac{J_{a-1}(\sqrt{x})}{\sqrt{y}}K_{\mathrm{Bessel}}^{(a-1)}(x,y)\left(1-2\mathbf{J}_{a-1}(\sqrt{x})\right)
F(x)F(y)dxdy\nonumber\\
&+&\frac{1}{32}\int_{0}^{\infty}\int_{0}^{\infty}\frac{J_{a-1}(\sqrt{x})}{\sqrt{x}}
\left(1-2\mathbf{J}_{a-1}(\sqrt{y})\right)\left(L^{(a-1)}(x,y)-L^{(a-1)}(y,x)\right)\no\\
&&F(x)F'(y)dxdy-\frac{1}{32}\int_{0}^{\infty}\int_{0}^{\infty}L^{(a-1)}(x,y)L^{(a-1)}(y,x)F'(x)F'(y)dxdy\nonumber\\
&-&\frac{1}{64}\int_{0}^{\infty}\int_{0}^{\infty}\left(1-2\mathbf{J}_{a-1}(\sqrt{x})\right)\left(1-2\mathbf{J}_{a-1}(\sqrt{y})\right)L^{(a-1)}(x,y)\no\\
&&F'(x)F'(y)dxdy+\frac{1}{4}\int_{0}^{\infty}L^{(a-1)}(x,x)F(x)F'(x)dx+O(N^{-1}),\nonumber
\eea
\es
where $\tilde{\mu}_{N}^{(\mathrm{JUE},a-1)}$ and $\tilde{\mathcal{V}}_{N}^{(\mathrm{JUE},a-1)}$ are given by (\ref{juem2}) and (\ref{juev2}) with $a$ replaced by $a-1$, and $\mathbf{J}_{a-1}(x)$ and $L^{(a-1)}(x,y)$ are defined by (\ref{jx}) and (\ref{lxy}) respectively.
\end{theorem}

\section{Jacobi Orthogonal Ensemble (JOE)}
\subsection{Finite $N$ Case for the MGF in JOE}
In the JOE case, we take the weight $w(x)$ to be $w(x)=(1-x)^{a/2}(1+x)^{b/2},\;\;x\in [-1,1],\;a, b>-2$ for convenience.
We assume that $N$ is even. The authors \cite{Min201601} expressed $G_{N}^{(1)}(f)$ as a Fredholm determinant based on \cite{Dieng,Tracy1998}. Let
$$
\psi_{j}^{(1)}(x):=\pi_{j}(x)w(x),\;\;j=0,1,2,\ldots,
$$
where $\pi_{j}(x)$ is an arbitrary polynomial of degree $j$, and
$$
M^{(1)}:=\left(\int_{-1}^{1}\psi_{j}^{(1)}(x)\varepsilon\psi_{k}^{(1)}(x)dx\right)_{j,k=0}^{N-1}
$$
with its inverse denoted by
$$
\left(M^{(1)}\right)^{-1}=:(\nu_{jk})_{j,k=0}^{N-1}.
$$
It was shown in \cite{Min201601} that
\be\label{gn1}
\left[G_{N}^{(1)}(f)\right]^{2}
=\det\left(I+K_{N}^{(1)}(f^{2}+2f)-K_{N}^{(1)}\varepsilon f'-K_{N}^{(1)}f\varepsilon f'\right),
\ee
where $K_{N}^{(1)}$ is the integral operator with kernel
\be\label{kn1}
K_{N}^{(1)}(x,y):=\sum_{j,k=0}^{N-1}\nu_{jk}\varepsilon\psi_{j}^{(1)}(x)\psi_{k}^{(1)}(y).
\ee
We also require that $f\in C^{1}[-1,1]$ and vanishes at the endpoints $\pm 1$.

Similarly as the discussions in Section 3.1,
we can choose a special $\psi_{j}^{(1)}$ to simplify $M^{(1)}$ as the direct sum of $N/2$ copies of
\bs
$
\begin{pmatrix}
0&1\\
-1&0
\end{pmatrix}
$\es. Let
$$
\psi_{2j+1}^{(1)}(x)=\frac{d}{dx}\left[(1-x^2)\varphi_{2j}^{(1)}(x)\right],\qquad \psi_{2j}^{(1)}(x)=\varphi_{2j}^{(1)}(x),\;\;j=0,1,2,\ldots,
$$
where $\varphi_{j}^{(1)}(x)$ is given by
$$
\varphi_{j}^{(1)}(x)=\frac{P_{j}^{(a+1,b+1)}(x)}{\sqrt{h_{j}^{(a+1,b+1)}}}(1-x)^{a/2}(1+x)^{b/2},
$$
and $P_{j}^{(a+1,b+1)}(x), j=0,1,\ldots$ are the Jacobi polynomials with the orthogonality condition
$$
\int_{-1}^{1}P_{j}^{(a+1,b+1)}(x)P_{k}^{(a+1,b+1)}(x)(1-x)^{a+1}(1+x)^{b+1} dx=h_{j}^{(a+1,b+1)}\delta_{jk}.
$$
It is easy to see that $\psi_{j}^{(1)}(x)$ is equal to $(1-x)^{a/2}(1+x)^{b/2}$ multiplied by a polynomial of degree $j$. Moreover, $M^{(1)}$ is computed to be
$$
M^{(1)}=\begin{pmatrix}
0&1&0&0&\cdots&0&0\\
-1&0&0&0&\cdots&0&0\\
0&0&0&1&\cdots&0&0\\
0&0&-1&0&\cdots&0&0\\
\vdots&\vdots&\vdots&\vdots&&\vdots&\vdots\\
0&0&0&0&\cdots&0&1\\
0&0&0&0&\cdots&-1&0
\end{pmatrix}_{N\times N}.
$$
It follows that $\left(M^{(1)}\right)^{-1}=-M^{(1)}$, so $\nu_{2j,2j+1}=-1,\; \nu_{2j+1,2j}=1$ and $\nu_{jk}=0$ for other cases.

\begin{lemma}
We have
\be\label{kn1a}
K_N^{(1)}(x,y)=S_N^{(1)}(x,y)+C_N^{(1)}\varepsilon\varphi_{N}^{(1)}(x)
\varphi_{N-1}^{(1)}(y),
\ee
where
$$
C_N^{(1)}=\sqrt{\frac{N(N+a+1)(N+b+1)(N+a+b+2)}{(2N+a+b+1)(2N+a+b+3)}}
$$
and
$$
S_N^{(1)}(x,y)=\sum_{j=0}^{N-1}(1-x^2)\varphi_{j}^{(1)}(x)\varphi_{j}^{(1)}(y).
$$
\end{lemma}
\begin{proof}
According to (\ref{kn1}), we find that $K_{N}^{(1)}(x,y)$ equals
\bea\label{kn1e}
&&\sum_{j=0}^{\frac{N}{2}-1}\varepsilon\psi_{2j+1}^{(1)}(x)\psi_{2j}^{(1)}(y)-\sum_{j=0}^{\frac{N}{2}-1}\varepsilon\psi_{2j}^{(1)}(x)\psi_{2j+1}^{(1)}(y)\nonumber\\
&&=\sum_{j=0}^{\frac{N}{2}-1}(1-x^2)\varphi_{2j}^{(1)}(x)\varphi_{2j}^{(1)}(y)-\sum_{j=0}^{\frac{N}{2}-1}\varepsilon\varphi_{2j}^{(1)}(x)\left[(1-y^2)\varphi_{2j}^{(1)}(y)\right]'.
\eea
In view of the recurrence formulas for the Jacobi polynomials (\ref{rf1}) and (\ref{rf2}), we find
$$
\left[(1-y^2)\varphi_{2j}^{(1)}(y)\right]'=C_{2j}^{(1)}\varphi_{2j-1}^{(1)}(y)-C_{2j+1}^{(1)}\varphi_{2j+1}^{(1)}(y),
$$
where
$$
C_j^{(1)}:=\sqrt{\frac{j(j+a+1)(j+b+1)(j+a+b+2)}{(2j+a+b+1)(2j+a+b+3)}}.
$$
Then (\ref{kn1e}) becomes
\bea\label{com11}
K_{N}^{(1)}(x,y)&=&\sum_{j=0}^{\frac{N}{2}-1}(1-x^2)\varphi_{2j}^{(1)}(x)\varphi_{2j}^{(1)}(y)+C_{N}^{(1)}\varepsilon\varphi_{N}^{(1)}(x)\varphi_{N-1}^{(1)}(y)\no\\
&+&\sum_{j=1}^{\frac{N}{2}}\left[C_{2j-1}^{(1)}\varepsilon\varphi_{2j-2}^{(1)}(x)
-C_{2j}^{(1)}\varepsilon\varphi_{2j}^{(1)}(x)\right]\varphi_{2j-1}^{(1)}(y).
\eea
Using (\ref{rf1}) and (\ref{rf2}) again, we have
$$
\left[(1-x^2)\varphi_{2j-1}^{(1)}(x)\right]'=C_{2j-1}^{(1)}\varphi_{2j-2}^{(1)}(x)-C_{2j}^{(1)}\varphi_{2j}^{(1)}(x).
$$
It follows from Lemma \ref{lem} that
\be\label{com21}
(1-x^2)\varphi_{2j-1}^{(1)}(x)=\varepsilon\left[(1-x^2)\varphi_{2j-1}^{(1)}(x)\right]'=C_{2j-1}^{(1)}\varepsilon\varphi_{2j-2}^{(1)}(x)-C_{2j}^{(1)}\varepsilon\varphi_{2j}^{(1)}(x).
\ee
The combination of (\ref{com11}) and (\ref{com21}) produces
\bea
K_{N}^{(1)}(x,y)&=&\sum_{j=0}^{\frac{N}{2}-1}(1-x^2)\varphi_{2j}^{(1)}(x)\varphi_{2j}^{(1)}(y)
+\sum_{j=1}^{\frac{N}{2}}(1-x^2)\varphi_{2j-1}^{(1)}(x)\varphi_{2j-1}^{(1)}(y)\no\\
&+&C_{N}^{(1)}\varepsilon\varphi_{N}^{(1)}(x)\varphi_{N-1}^{(1)}(y)\no\\
&=&\sum_{j=0}^{N-1}(1-x^2)\varphi_{j}^{(1)}(x)\varphi_{j}^{(1)}(y)+C_{N}^{(1)}\varepsilon\varphi_{N}^{(1)}(x)\varphi_{N-1}^{(1)}(y).\no
\eea
The theorem is then established.
\end{proof}
\begin{theorem}\label{thm}
For the Jacobi orthogonal ensemble, we have
$$
\left[G_{N}^{(1)}(f)\right]^{2}=\det(I+T_{\mathrm{JOE}}),
$$
where
\bea\label{tjoe}
T_{\mathrm{JOE}}:&=&S_{N}^{(1)}(f^2+2f)-S_{N}^{(1)}\varepsilon f'-S_{N}^{(1)}f\varepsilon f'+C_N^{(1)}\left(\varepsilon\varphi_{N}^{(1)}\right)\otimes\varphi_{N-1}^{(1)}(f^2+2f)\nonumber\\
&+&C_N^{(1)}\left(\varepsilon\varphi_{N}^{(1)}\right)\otimes\left(\varepsilon\varphi_{N-1}^{(1)}\right) f'-C_N^{(1)}\left(\varepsilon\varphi_{N}^{(1)}\right)\otimes\varphi_{N-1}^{(1)}f\varepsilon f'.
\eea
\end{theorem}
\begin{proof}
Substituting (\ref{kn1a}) into (\ref{gn1}), we obtain the desired result.
\end{proof}

The amenable expression of $G_{N}^{(1)}(f)$ in the above theorem will allow us to study its large $N$ asymptotics in the next subsections by using the expansion formula
\be\label{logdet1}
\log\det(I+T_{\mathrm{JOE}})=\mathrm{Tr}\log(I+T_{\mathrm{JOE}})=\mathrm{Tr}\:T_{\mathrm{JOE}}-\frac{1}{2}\mathrm{Tr}\:T_{\mathrm{JOE}}^{2}
+\frac{1}{3}\mathrm{Tr}\:T_{\mathrm{JOE}}^{3}-\cdots.
\ee

\subsection{Scaling in the Bulk of the Spectrum in JOE}
Similarly as Theorem \ref{jse1} and \ref{jse2}, we have the following two theorems.
\begin{theorem}\label{joe1}
For $x, y\in\mathbb{R}$, we have as $N\rightarrow\infty$,
$$
\frac{1}{N}S_{N}^{(1)}\left(\frac{x}{N},\frac{y}{N}\right)
=K_{\mathrm{sine}}(x,y)+O(N^{-1}),
$$
uniformly for $x$ and $y$ in compact subsets of $\mathbb{R}$.
\end{theorem}

\begin{theorem}\label{joe2}
For $x\in\mathbb{R}$, we have as $N\rightarrow\infty$,
$$
\varphi_{N}^{(1)}\left(\frac{x}{N}\right)
=-\sqrt{\frac{2}{\pi}}\sin\left[\frac{1}{4}\left(\pi+2\pi a-2(2N+3+a+b)\arccos\frac{x}{N}\right)\right]+O(N^{-1}),
$$
$$
\varphi_{N-1}^{(1)}\left(\frac{x}{N}\right)
=-\sqrt{\frac{2}{\pi}}\sin\left[\frac{1}{4}\left(\pi+2\pi a-2(2N+1+a+b)\arccos\frac{x}{N}\right)\right]+O(N^{-1}),
$$
\bea
\varepsilon\varphi_{N}^{(1)}\left(\frac{x}{N}\right)
&=&\frac{1}{N\sqrt{2\pi}}\bigg\{\sin\left[\frac{1}{4}\left(\pi+2\pi a-2(2N+1+a+b)\arccos\frac{x}{N}\right)\right]\nonumber\\
&-&\sin\left[\frac{1}{4}\left(\pi+2\pi a-2(2N+5+a+b)\arccos\frac{x}{N}\right)\right]\bigg\}+O(N^{-2}),\nonumber
\eea
\bea
\varepsilon\varphi_{N-1}^{(1)}\left(\frac{x}{N}\right)
&=&\frac{1}{N\sqrt{2\pi}}\bigg\{\sin\left[\frac{1}{4}\left(\pi+2\pi a-2(2N-1+a+b)\arccos\frac{x}{N}\right)\right]\nonumber\\
&-&\sin\left[\frac{1}{4}\left(\pi+2\pi a-2(2N+3+a+b)\arccos\frac{x}{N}\right)\right]\bigg\}+O(N^{-2}).\nonumber
\eea
The error terms are uniform for $x$ in compact subsets of $\mathbb{R}$.
\end{theorem}

Using Theorem \ref{joe1} and \ref{joe2} to compute $\mathrm{Tr}\:T_{\mathrm{JOE}}$ and $\mathrm{Tr}\:T_{\mathrm{JOE}}^{2}$ from (\ref{tjoe}), and changing $f(x)$ to $f(Nx)$, we find that $\mathrm{Tr}\:T_{\mathrm{JOE}}$ and $\mathrm{Tr}\:T_{\mathrm{JOE}}^{2}$ equal
\bea
&&\int_{-\infty}^{\infty}K_{\mathrm{sine}}(x,x)\left(f^2(x)+2f(x)\right)dx\no\\
&-&\frac{1}{2}\int_{-\infty}^{\infty}
\left[\int_{-\infty}^{\infty}\left(1-2\chi_{(-\infty,x)}(y)\right)K_{\mathrm{sine}}(x,y)f(y)dy\right]f'(x)dx+O(N^{-1})\no
\eea
and
\bs
\bea
&&\int_{-\infty}^{\infty}\int_{-\infty}^{\infty}K_{\mathrm{sine}}^2(x,y)\left(f^2(x)+2f(x)\right)\left(f^2(y)+2f(y)\right)dxdy\nonumber\\
&-&\int_{-\infty}^{\infty}\int_{-\infty}^{\infty}K_{\mathrm{sine}}(x,y)\left[\int_{-\infty}^{\infty}\left(1-2\chi_{(-\infty,x)}(z)\right)K_{\mathrm{sine}}(y,z)f(z)dz\right]
\nonumber\\
&&f'(x)\left(f^2(y)+2f(y)\right)dxdy+\frac{1}{\pi^2}\int_{-\infty}^{\infty}\int_{-\infty}^{\infty}\mathrm{Si}^2(x-y)f'(x)f'(y)(2f(y)+1)dxdy\nonumber\\
&+&\frac{1}{\pi}\int_{-\infty}^{\infty}\int_{-\infty}^{\infty}\mathrm{Si}(x-y)\left[\int_{-\infty}^{\infty}\left(1-2\chi_{(-\infty,x)}(z)\right)
K_{\mathrm{sine}}(y,z)f(z)dz\right]f'(x)f'(y)dxdy\nonumber\\
&+&\frac{1}{4}\int_{-\infty}^{\infty}\int_{-\infty}^{\infty}\left[\int_{-\infty}^{\infty}\left(1-2\chi_{(-\infty,x)}(z)\right)
K_{\mathrm{sine}}(y,z)f(z)dz\right]\nonumber\\
&&\left[\int_{-\infty}^{\infty}\left(1-2\chi_{(-\infty,y)}(u)\right)
K_{\mathrm{sine}}(x,u)f(u)du\right]f'(x)f'(y)dxdy+O(N^{-1})\nonumber
\eea
\es%
respectively,
where we have used integration by parts and $\chi_{J}(x)$ is the characteristic function of the interval $J$, i.e., $\chi_{J}(x)=1$ for $x\in J$ and $0$ otherwise.
\begin{remark}
We assume that $f(\cdot)$ is smooth and sufficiently decreasing at $\pm\infty$ to make the integrals well-defined.
\end{remark}
With the relation of $f(x)$ and $F(x)$ in (\ref{fx}), we find from (\ref{logdet1}) that
\bs
\bea
\log\det(I+T_{\mathrm{JOE}})&=&-2\lambda\int_{-\infty}^{\infty}K_{\mathrm{sine}}(x,x)F(x)dx+\frac{\lambda^2}{2}\bigg\{4\int_{-\infty}^{\infty}K_{\mathrm{sine}}(x,x)F^2(x)dx
\nonumber\\
&-&\int_{-\infty}^{\infty}\left[\int_{-\infty}^{\infty}\left(1-2\chi_{(-\infty,x)}(y)\right)
K_{\mathrm{sine}}(x,y)F(y)dy\right]F'(x)dx\nonumber\\
&-&4\int_{-\infty}^{\infty}\int_{-\infty}^{\infty}K_{\mathrm{sine}}^2(x,y)F(x)F(y)dxdy\nonumber\\
&-&\frac{1}{\pi^2}\int_{-\infty}^{\infty}\int_{-\infty}^{\infty}\mathrm{Si}^2(x-y)F'(x)F'(y)dxdy\bigg\}+\cdots+O(N^{-1}).\nonumber
\eea
\es
Taking account of the fact that $\log G_{N}^{(1)}(f)=\frac{1}{2}\log\det(I+T_{\mathrm{JOE}})$ from Theorem \ref{thm}, we have the following heuristic result.
\begin{theorem}
Letting $\mu_{N}^{(\mathrm{JOE})}$ and $\mathcal{V}_{N}^{(\mathrm{JOE})}$ be the mean and variance of the scaled linear statistics
$\sum_{j=1}^{N}F(Nx_j)$, we have as $N\rightarrow\infty$,
$$
\mu_{N}^{(\mathrm{JOE})}=\mu_{N}^{(\mathrm{JUE})}+O(N^{-1}),
$$
\bea
\mathcal{V}_{N}^{(\mathrm{JOE})}
&=&2\mathcal{V}_{N}^{(\mathrm{JUE})}-\frac{1}{2}\int_{-\infty}^{\infty}\left[\int_{-\infty}^{\infty}\left(1-2\chi_{(-\infty,x)}(y)\right)
K_{\mathrm{sine}}(x,y)F(y)dy\right]F'(x)dx\nonumber\\
&-&\frac{1}{2\pi^2}\int_{-\infty}^{\infty}\int_{-\infty}^{\infty}\mathrm{Si}^2(x-y)F'(x)F'(y)dxdy+O(N^{-1}),\nonumber
\eea
where $\mu_{N}^{(\mathrm{JUE})}$ and $\mathcal{V}_{N}^{(\mathrm{JUE})}$ are given by (\ref{juem}) and (\ref{juev}), respectively.
\end{theorem}

\subsection{Scaling at the Edge of the Spectrum in JOE}
Similarly as Theorem \ref{jse3} and \ref{jse4}, we have the following results.
\begin{theorem}\label{joe3}
For $x, y\in\mathbb{R}^{+}$, we have as $N\rightarrow\infty$,
$$
\frac{1}{2N^2}S_{N}^{(1)}\left(1-\frac{x}{2N^2},1-\frac{y}{2N^2}\right)
=\sqrt{\frac{x}{y}}\:K_{\mathrm{Bessel}}^{(a+1)}(x,y)+O(N^{-1}),
$$
where $K_{\mathrm{Bessel}}^{(a+1)}(x,y)$ is the Bessel kernel of order $a+1$ given by
$$
K_{\mathrm{Bessel}}^{(a+1)}(x,y)=\frac{J_{a+1}(\sqrt{x})\sqrt{y}J_{a+1}'(\sqrt{y})-J_{a+1}'(\sqrt{x})\sqrt{x}J_{a+1}(\sqrt{y})}{2(x-y)}.
$$
The error term is uniform for $x$ and $y$ in compact subsets of $\mathbb{R}^{+}$.
\end{theorem}

\begin{theorem}\label{joe4}
For $x\in\mathbb{R}^{+}$, we have as $N\rightarrow\infty$,
$$
\varphi_{N}^{(1)}\left(1-\frac{x}{2N^2}\right)
=N^{3/2}\:\frac{J_{a+1}(\sqrt{x})}{\sqrt{x}}+O(N^{1/2}),
$$
$$
\varphi_{N-1}^{(1)}\left(1-\frac{x}{2N^2}\right)
=N^{3/2}\:\frac{J_{a+1}(\sqrt{x})}{\sqrt{x}}+O(N^{1/2}),
$$
$$
\varepsilon\varphi_{N}^{(1)}\left(1-\frac{x}{2N^2}\right)
=2^{-1}N^{-1/2}\left(1-2\mathbf{J}_{a+1}(\sqrt{x})\right)+O(N^{-3/2}),
$$
$$
\varepsilon\varphi_{N-1}^{(1)}\left(1-\frac{x}{2N^2}\right)
=2^{-1}N^{-1/2}\left(1-2\mathbf{J}_{a+1}(\sqrt{x})\right)+O(N^{-3/2}),
$$
where
\be\label{jx1}
\mathbf{J}_{a+1}(x)=\int_{0}^{x}J_{a+1}(t)dt.
\ee
The error terms are uniform for $x$ in compact subsets of $\mathbb{R}^{+}$.
\end{theorem}

Using Theorem \ref{joe3} and \ref{joe4} to compute $\mathrm{Tr}\:T_{\mathrm{JOE}}$ and $\mathrm{Tr}\:T_{\mathrm{JOE}}^{2}$ from (\ref{tjoe}), and changing $f(x)$ to $f(2N^2(1-x))$, we obtain the next theorem following the similar heuristic procedure in Section 4.2. (We assume that $f(\cdot)$ is smooth and sufficiently decreasing at infinity.)
\begin{theorem}
Denoting by $\tilde{\mu}_{N}^{(\mathrm{JOE})}$ and $\tilde{\mathcal{V}}_{N}^{(\mathrm{JOE})}$ the mean and variance of the linear statistics
$\sum_{j=1}^{N}F(2N^2(1-x_j))$, we have as $N\rightarrow\infty$,
$$
\tilde{\mu}_{N}^{(\mathrm{JOE})}=\tilde{\mu}_{N}^{(\mathrm{JUE},a+1)}+\frac{1}{4}\int_{0}^{\infty}L^{(a+1)}(x,x)F'(x)dx+O(N^{-1}),
$$
\bs
\bea
\tilde{\mathcal{V}}_{N}^{(\mathrm{JOE})}
&=&2\tilde{\mathcal{V}}_{N}^{(\mathrm{JUE},a+1)}+\frac{1}{2}\int_{0}^{\infty}L^{(a+1)}(x,x)F(x)F'(x)dx\nonumber\\
&-&\frac{1}{2}\int_{0}^{\infty}\left[\int_{0}^{\infty}\left(1-2\chi_{(0,x)}(y)\right)\sqrt{\frac{x}{y}}\:K_{\mathrm{Bessel}}^{(a+1)}(x,y)F(y)dy\right]F'(x)dx\nonumber\\
&-&\frac{1}{16}\int_{0}^{\infty}\left[\int_{0}^{\infty}\left(1-2\chi_{(0,x)}(y)\right)\frac{J_{a+1}(\sqrt{y})}{\sqrt{y}}F(y)dy\right]
\frac{J_{a+1}(\sqrt{x})}{\sqrt{x}}F(x)dx\nonumber\\
&-&\frac{1}{2}\int_{0}^{\infty}\int_{0}^{\infty}\frac{J_{a+1}(\sqrt{x})}{\sqrt{y}}K_{\mathrm{Bessel}}^{(a+1)}(x,y)\left(1-2\mathbf{J}_{a+1}(\sqrt{y})\right)F(x)F(y)dxdy\nonumber\\
&+&\frac{1}{8}\int_{0}^{\infty}\int_{0}^{\infty}\frac{J_{a+1}(\sqrt{x})}{\sqrt{x}}\left(1-2\mathbf{J}_{a+1}(\sqrt{y})\right)\left(L ^{(a+1)}(x,y)-L^{(a+1)}(y,x)\right)\no\\
&&F(x)F'(y)dxdy-\frac{1}{8}\int_{0}^{\infty}\int_{0}^{\infty}L^{(a+1)}(x,y)L^{(a+1)}(y,x)F'(x)F'(y)dxdy\nonumber\\
&-&\frac{1}{16}\int_{0}^{\infty}\int_{0}^{\infty}\left(1-2\mathbf{J}_{a+1}(\sqrt{x})\right)\left(1-2\mathbf{J}_{a+1}(\sqrt{y})\right)L^{(a+1)}(x,y)F'(x)F'(y)dxdy\no\\
&-&\int_{0}^{\infty}\int_{0}^{\infty}\sqrt{\frac{x}{y}}\:K_{\mathrm{Bessel}}^{(a+1)}(x,y)L^{(a+1)}(x,y)F'(x)F(y)dxdy+O(N^{-1}),\nonumber
\eea
\es
where
$$
L^{(a+1)}(x,y)=\int_{0}^{x}\sqrt{\frac{y}{z}}\:K_{\mathrm{Bessel}}^{(a+1)}(y,z)dz-\int_{x}^{\infty}\sqrt{\frac{y}{z}}\:K_{\mathrm{Bessel}}^{(a+1)}(y,z)dz,
$$
$\mathbf{J}_{a+1}(x)$ is defined in (\ref{jx1}) and $\tilde{\mu}_{N}^{(\mathrm{JUE},a+1)}$ and $\tilde{\mathcal{V}}_{N}^{(\mathrm{JUE},a+1)}$ are given by (\ref{juem2}) and (\ref{juev2}) with $a$ replaced by $a+1$.
\end{theorem}
\begin{remark}
We have used integration by parts to simplify the results in the above theorem.
\end{remark}

\begin{acknowledgement}
We would like to thank Professor E. Basor, Professor A. B\"{o}ttcher and Professor T. Ehrhardt for giving many useful comments.
The work of C. Min was partially supported by the National Natural Science Foundation of China under grant number 12001212, by the Fundamental Research Funds for the Central Universities under grant number ZQN-902 and by the Scientific Research Funds of Huaqiao University under grant number 17BS402. The work of
Y. Chen was partially supported by the Macau Science and Technology Development Fund under grant number FDCT 0079/2020/A2.
\end{acknowledgement}

\end{document}